\documentclass[a4paper, 10pt]{article}

\usepackage[T5,T1]{fontenc}
\usepackage[utf8]{inputenc}

\usepackage{amsmath}
\usepackage{amsthm}
\usepackage{amssymb}

\usepackage{csquotes} 

\usepackage[a4paper,includeheadfoot,margin=3cm]{geometry}
\usepackage{enumerate}
\usepackage[hidelinks]{hyperref}

\usepackage{cleveref}

\theoremstyle{plain}
\newtheorem{theorem}{Theorem}

\theoremstyle{definition}

\newtheorem{claim}[theorem]{Claim}

\newcommand\bbN{\mathbb{N}}
\newcommand\angles[1]{\langle #1 \rangle}

\begin{document}

\title{A finer reparameterisation theorem for MSO and FO queries on strings}
\author{{\fontencoding{T5}\selectfont Lê Thành Dũng (Tito) Nguyễn} \and Paweł Parys}
\date{}
\maketitle

\begin{abstract}
  We show a theorem on monadic second-order k-ary queries on finite words. It may be illustrated by the following example: if the number of results of a query on binary strings is O(number of 0s $\times$ number of 1s), then each result can be MSO-definably identified from a 0-position, a 1-position and some finite data.

  Our proofs also handle the case of first-order logic / aperiodic monoids. Thus we can state and prove the folklore theorem that dimension minimisation holds for first-order string-to-string interpretations.
\end{abstract}

For an MSO query given by a formula $\varphi(x_1,\dots,x_k)$ (using first-order
free variables $x_i$) on words over a finite alphabet $\Sigma$, and
$w\in\Sigma^*$, we write
\[  \#\varphi(w) = \mathrm{Card}(\{\vec{i} \in \{1,\dots,|w|\}^k \mid w
  \models \varphi(\vec{i})\})\]
Our main result (originally conjectured by Thomas Colcombet in personal communication) is:
\begin{theorem}\label{thm:main}
  Let $\varphi(x_1,\dots,x_k)$ and $\eta_1(x),\dots,\eta_\ell(x)$ be MSO (resp.\ FO)
  formulas such that
  \[\#\varphi(w) = O(\#\eta_1(w) \times \dots \times
    \#\eta_\ell(w))\]
  Then there exists an MSO (resp.\ FO) formula $\psi(x_1,\dots,x_k,y_1,\dots,y_\ell)$,
  which defines for each word a \emph{total functional} relation
  \[ \{\vec{i}  \mid w
    \models \varphi(\vec{i})\} \to \{\vec{j} \mid w \models \eta_1(j_1) \land
    \dots \land \eta_\ell(j_\ell)\} \]
  where \emph{each $\vec{j}$ has $O(1)$ many preimages $\vec{i}$}.
\end{theorem}

For MSO, the special case $\eta_1(x) = \dots = \eta_\ell(x) = \mathtt{true}$ was first stated by Bojánczyk in~\cite[Lemma~6.2]{PolyregSurvey}, with the proof of a slight variant appearing in~\cite{PolyregularGrowth}. This proof uses Imre Simon's factorisation forest theorem. Our proof of \Cref{thm:main} is mainly based on ideas from~\cite{PolyregularGrowth}, but with choices of exposition heavily inspired from~\cite{gaetanPhD}. The case where the $\eta_i$ are trivial has been extended:
\begin{itemize}
  \item to trees in~\cite[Thm.~1.3]{StructPoly} using arguably simpler tools (combinatorics of weighted automata), but we have not found a way to apply this alternative approach to derive~\Cref{thm:main};
  \item to countable linear orders in~\cite{RabinovichWoLLIC} (see also~\cite{RabinovichICALP} on MSO set queries).
\end{itemize}
For FO, even the \enquote{$\eta_i$ trivial} case has not appeared previously in the literature; it is not just a corollary of the MSO case since we need to ensure that $\psi$ is in FO.\@ But the proof scheme from~\cite{PolyregularGrowth} still applies, using an aperiodic version of factorisation forests.\footnote{This observation comes from Bojańczyk (personal communication).} Our proof of \Cref{thm:main} establishes the MSO and FO cases simultaneously.

The original motivation of~\cite[Lemma~6.2]{PolyregSurvey} was to prove a dimension minimisation theorem for string-to-string MSO interpretations~\cite[Theorem~6.1]{PolyregSurvey} (see also~\cite[Theorem~1.5]{StructPoly} for trees). The reduction to a result on MSO queries performed in~\cite{PolyregSurvey,StructPoly} is an elementary syntactic manipulation that also works for FO.\@ Therefore, thanks to the FO and \enquote{$\eta_i$ trivial} case of \Cref{thm:main}, dimension minimisation holds for FO interpretations:
\begin{theorem}
  An FO interpretation that defines a string-to-string function $f$ such that $|f(w)| = O(|w|^\ell)$ can be effectively translated to an \emph{$\ell$-dimensional} FO interpretation that also defines $f$. 
\end{theorem}

\section{Proof of the main theorem}

\paragraph{Recognizing queries.}

Consider now a query $\varphi$ and assume that $\varphi(x_1,\dots,x_k)
\Rightarrow x_1 < \dots < x_k$ holds for all words, w.l.o.g. We can then denote
$aaabbabba \models \varphi(3,5,9)$ as
$aa\underline{a}b\underline{b}abb\underline{a} \models \varphi$.

We say that a monoid morphism $\mu\colon\Sigma^*\to M$ recognizes $\varphi$ when the
images of the maximal infixes without distinguished position -- in the above
example, $\mu(aa)$, $\mu(b)$, $\mu(abb)$ and $\mu(\varepsilon)$ -- together with
the values of the distinguished letters ($a,b,a$ above) suffice to determine
whether $\varphi$ is satisfied. A query is recognizable by a morphism to a finite (resp.\ finite and aperiodic) monoid if and only if it is
MSO-definable (resp.\ FO-definable).

\paragraph{Factorization forests.}

Let $\mu\colon\Sigma^* \to M$ be a morphism to a finite monoid. We define a $\mu$-\emph{forest} for $w\in\Sigma^+$ by induction as:
\begin{itemize}
\item either the leaf $w$ when $w\in\Sigma$
\item or the node $\angles{f_1}\dots\angles{f_n}$ whose children $f_1,\dots,f_n$
  are forests for $w_1,\dots,w_n$ such that $w = w_1 \cdot \ldots \cdot w_n$
        with $n\geq2$, and with the condition that if $n\geq3$ then:
        \begin{itemize}
          \item $\mu(w_1) = \dots = \mu(w_n)$;
          \item $\mu(w_1)^{|M|} = \mu(w_1)^{|M|+1}$ --- beware, this final condition is a bit idiosyncratic, and meant to unify the MSO and FO cases.
        \end{itemize}
\end{itemize}
The \emph{height} is the maximum nesting of $\angles{-}$.
\begin{theorem}
  There exists a \emph{rational function} $\Sigma^* \to (\Sigma\cup\{\langle,\rangle\})^*$ that produces for each input word (the string representation of) some $\mu$-forest for that word of bounded height, with the expected origin semantics.
  Furthermore, when $M$ is aperiodic, this function can be taken to be FO-rational.
\end{theorem}
\begin{proof}
  In the non-aperiodic case, the Factorisation Forest Theorem states the existence of a $\mu$-forest of bounded height with the last item replaced by the stronger property that $\mu(w_1)$ is \emph{idempotent}, and a well-known refinement (see e.g.~\cite[Theorem~2.21]{gaetanPhD}) gives us a rational function that computes it.

  In the aperiodic case, an FO-rational function can produce forests without this idempotence property according to~\cite[Lemma~6.5]{polyregular}\footnote{For a proof, see~\cite[Appendix~B]{msoInterpretations}.} and then $\mu(w_1)^{|M|} = \mu(w_1)^{|M|+1}$ follows from the definition of the aperiodicity of $M$.
\end{proof}

\paragraph{Navigating in forests.}

Changing the definitions of~\cite[Section~2.3]{gaetanPhD} slightly to fit our notion of $\mu$-forest, we say that a node $\mathfrak{m}$ in a
forest \emph{observes} a node $\mathfrak{n}$ when $\mathfrak{n}$ is a sibling at distance at most $|M|$ (when ordering the siblings from left to right) of an ancestor of $\mathfrak{m}$. (Here, any node is its own sibling at distance 0, and it is also its own ancestor.) A node is \emph{iterable} when it has at least $|M|$ left siblings and at least $|M|$ right siblings (not including itself).

Original terminology: the \emph{anchor} of a leaf is either its lowest iterable
ancestor, if it has any, or the root otherwise. (This is
related to the \enquote{frontiers} in~\cite[\S2.3]{gaetanPhD}.) When the anchor
of a leaf~$i$ observes the anchor of a leaf $j$, we say that $i$
\emph{points to} $j$.

Fix a $k$-ary MSO (resp.\ FO) query $\varphi$ and unary queries $\eta_1,\dots,\eta_\ell$, a
morphism $\mu\colon\Sigma^* \to M$ recognizing all those queries, and a rational (resp.\ FO-rational)
function computing factorization forests of bounded height for $\mu$. For
$w\in\Sigma^*$, we speak of \emph{the} forest of $w$ to refer to the one
returned by the aforementioned function.

Let $(w,i_1,\dots,i_k)$ such that $w \models \varphi(i_1,\dots,i_k)$. Let us
consider the following graph: the vertices are $\{1,\dots,k\}$ and there is an
edge from $p$ to $q$ when $i_p$ points to $i_q$.
\begin{claim}
  Let $S$ be a strongly connected component of this graph. Then all the nodes in the set
  $\mathrm{anchors}(S)=\{\mathrm{anchor}(i_p) \mid p \in S\}$ is are siblings in
  the forest of $w$, and any two consecutive members of that set are at distance at most $|M|$. (This includes the case where it is the singleton containing
  the root.)
\end{claim}
\begin{proof}
  If $i_p$ points to $i_q$ then $\mathrm{height}(\mathrm{anchor}(i_p))
  \leqslant \mathrm{height}(\mathrm{anchor}(i_q))$. Thus, we first deduce that
  all anchors have the same height. Thus, in this case, $i_p$ can only point to
  $i_q$ if $\mathrm{anchor}(i_p)$ is a (not necessarily proper) sibling of
  $\mathrm{anchor}(i_q)$. We conclude by again using connectedness.
\end{proof}
This claim allows us to define $\overline{\mathrm{anchors}}(S)$ as the set consisting of the nodes from $\mathrm{anchors}(S)$ and the siblings between them in the left-to-right order --- note that all these siblings are iterable, unless $S$ is the singleton containing the root.

Let us call $\mathcal{M}$ the set of strongly connected components that are
minimal, i.e.\ that have no incoming edge. For each $S\in \mathcal{M}$, we
consider the infix $\mathrm{block}(S)$ of $w$ obtained by taking the leaves that
descend from the nodes in $\overline{\mathrm{anchors}}(S)$.
\begin{claim}
  The infixes $\mathrm{block}(S)$ for $S\in\mathcal{M}$ are non-overlapping in $w$.
\end{claim}
\begin{proof}
  If they were overlapping, then there would be some leaf that is a common descendant of some two nodes $\alpha_1 \in \overline{\mathrm{anchors}}(S_1)$ and $\alpha_2 \in \overline{\mathrm{anchors}}(S_2)$ for $S_1 \neq S_2$ and $S_1,S_2 \in \mathcal{M}$. This would mean that one of $\alpha_1$ and $\alpha_2$ is an ancestor of the other.
  \begin{itemize}
    \item Suppose first that they are equal. Then $\overline{\mathrm{anchors}}(S_1) \cap \overline{\mathrm{anchors}}(S_2) \neq \varnothing$. This can only happen if some node in $\mathrm{anchors}(S_1)$ and some other node in $\mathrm{anchors}(S_2)$ are at distance at most $|M|$, contradicting the fact that $S_1$ and $S_2$ are distinct strongly connected components. 
    \item We may now assume w.l.o.g.\ that $\alpha_1$ is a strict ancestor of $\alpha_2$ --- and therefore of all its siblings, in particular of some $\beta_2 \in \mathrm{anchors}(S_2)$. There also exists $\beta_1 \in \mathrm{anchors}(S_1)$ that is a sibling at distance at most $|M|$ of $\alpha_1$. By definition, $\beta_2$ observes $\beta_1$, contradicting the minimality of the component $S_1$. \qedhere
  \end{itemize}
\end{proof}
Let $\mathcal{S \subseteq M}$ be any non-empty subset. We define
$w^{\mathcal{S}}_n$ as the word obtained from $w$ by pumping $n$ times all the infixes
$\mathrm{block}(S)$ for $S\in\mathcal{S}$. Clearly, $|w^{\mathcal{S}}_n| = O(n)$.
\begin{claim}
  $\#\varphi(w^{\mathcal{S}}_n) \geqslant n^{|\mathcal{S}|}$.
\end{claim}
\begin{proof}
  For each pumped block, let us choose one of its copies in $w^{\mathcal{S}}_n$; there are $n^{|\mathcal{S}|}$ possible choices.
  Let $j_{1},\dots,j_{k}$ be positions in $w^{\mathcal{S}}_n$:
  \begin{itemize}
    \item that correspond to $i_{1},\dots,i_{k}$ in $w$, according to the intuitive origin semantics of pumping;
    \item such that if $i_{x}$ is in a pumped block, then $j_{x}$ is in its chosen copy.
  \end{itemize}
  We claim that $w^{\mathcal{S}}_n \models \varphi(j_{1},\dots,j_{k})$.
          Idea: use the fact that the nodes in $\overline{\mathrm{anchors}}(S)$ for $S \in \mathcal{S}$ are iterable, plus the standard
  argument that one can reconstruct the value $\mu(w[j_{x}+1 \dots j_{x+1} - 1])$ from the nodes
  observed by $j_{x}$ and $j_{x+1}$ in a forest for $w^{\mathcal{S}}_n$ deduced by pumping.
\end{proof}
Furthermore, for each query $\eta_p$ ($1\leqslant p\leqslant \ell$), since $\mu$
recognizes it, it can also be evaluated using the $\mu$-forest, so we have in the
above pumping construction:
\[ (\text{no positions in the pumped infixes}\ \mathrm{block}(S)\ \text{satisfy}\ \eta_p) \implies
  \#\eta_p(w^{\mathcal{S}}_n) = O(1)\]
Let $P(\mathcal{S}) = \{p \in \{1,\dots,\ell\} \mid \exists S \in \mathcal{S} :
\text{some position in}\ \mathrm{block}(S)\ \text{satisfies}\ \eta_p\}$. We then have
\[ \#\eta_1(w^{\mathcal{S}}_n) \times \dots \times \#\eta_\ell(w^{\mathcal{S}}_n) = O(n^{|P(\mathcal{S})|})
  \quad\text{therefore}\quad |\mathcal{S}| \leqslant |P(\mathcal{S})|\]
Since this holds for all non-empty $\mathcal{S \subseteq M}$, by Hall's marriage
theorem, there exists an injection $\mathcal{M} \hookrightarrow
\{1,\dots,\ell\}$ that maps each $S\in \mathcal{M}$ to a $p_{S}$ such that the
infix $\mathrm{block}(S)$ in $w$ contains some position $j$ such that $w \models
\eta_{p_S}(j)$. We can fix a choice of injection $S \mapsto p_S$, e.g.\ the
lexicographically smallest one. Let us define, for $p\in\{1,\dots,\ell\}$,
\[ j_p =
  \begin{cases}
    \text{the leftmost position of}\ w\ \text{inside}\ \mathrm{block}(S)\ \text{satisfying}\
    \eta_{p},\ \text{when}\ p = p_S \\
    1\ \text{when there is no such}\ S
  \end{cases}
\]
\begin{claim}
  The functional relation that maps $(w,i_1,\dots,i_k)$ to $(j_1,\dots,j_\ell)$,
  according to the above recipe, can be uniformly defined over all words by some MSO formula
  $\psi(x_1,\dots,x_k,y_1,\dots,y_k)$. If $M$ is aperiodic then we can have $\psi$ in FO.
\end{claim}
\begin{proof}
  The formula $\psi$ has to compute $\mathcal{M}$, which is doable by an FO
  query on the forest. We combine this with the fact that FO queries on the
  output of a rational (resp.\ FO-rational) function can be pulled back to MSO (resp.\ FO) queries on the input.
\end{proof}

We now show that the number of preimages by this function is bounded. For each
$S$, the anchor of $j_{p_S}$ must be equal to or below the anchor of some $i_q$ where $q
\in S$; therefore, $j_{p_S}$ points to $i_q$. Since $\mathcal{M}$ consists of all
minimal strongly connected components, each $i_r$ for $r \in \{1,\dots,\ell\}$
is reachable by a path of length at most $k$ in the graph for the
\enquote{points to} relation whose vertices are $\{1,\dots,|w|\}$ (note that the
graph that we considered before can be seen as an induced subgraph), by some
$j_{p_S}$ where $S\in\mathcal{M}$. To conclude, recall that:
\begin{claim}
  Over forests of bounded height, a leaf points to a bounded number of other
  leaves.
\end{claim}
\begin{proof}
  All the arguments (but not the exact statement) may be found in~\cite[Section~2.3]{gaetanPhD}.
\end{proof}

\section{Counterexample to a tempting generalisation}

For $w\in\{a,b\}^*$, we have
\[ \#(a(x_1)\land b(x_2))(w) = \#a(w) \times \#b(w) \leqslant \#a(w)^2 + \#b(w)^2 = \#((a(y_1)\land
  a(y_2))\lor(b(y_1)\land b(y_2)))(w) \]
And yet:

\begin{theorem}
  There does not exist any MSO formula $\psi(x_1,x_2,y_1,y_2)$,
  which defines for each word a total functional relation
  \[ \{\vec{i}  \mid w
    \models a(i_1) \land b(i_2) \} \to \{\vec{j} \mid w[j_1] = w[j_2]\} \]
  where each $\vec{j}$ has $O(1)$ many preimages $\vec{i}$.
\end{theorem}
\begin{proof}
  For the sake of contradiction, assume that such a $\psi$ exists, with the bound $N\in\bbN$ on the number of preimages.
  Let
  \begin{align*}
    \psi'_i(x_1,x_2,y_1,y_2) &= (x_1,x_2)\ \text{is the $i$-th pair such that}\ \psi(\vec{x},\vec{y})  \\
    \varphi'_i(y_1,y_2) &= ((a(y_1)\land a(y_2))\lor(b(y_1)\land b(y_2)) \land \lnot (\exists \vec{x}.\, \psi'_i(\vec{x},\vec{y}))
  \end{align*}
  (ordering the pairs in e.g.\ lexicographic order). Let
  \[ f\colon w \in \{a,b\}^* \mapsto \#\varphi'_1(w) + \dots + \#\varphi'_N(w)\]
  By definition, $f$ is an \emph{$\bbN$-polyregular function}, cf.~\cite[Definition~5.10]{gaetanPhD}. Furthermore,
  \[ \forall w \in \{a,b\}^*,\quad f(w) = N(\#a(w)^2 + \#b(w)^2) - \#a(w) \times \#b(w) \]
  Since the polynomial $P(X,Y) = N(X^2+Y^2) - XY$ has a maximal monomial $-XY$ (for divisibility) with a negative coefficient, this contradicts~\cite[Theorem~17]{Aliaume} on commutative $\bbN$-polyregular functions.  
\end{proof}

\bibliographystyle{alphaurl}
\bibliography{bi}

\begin{thebibliography}{Rab26b}

\bibitem[BKL19]{msoInterpretations}
Mikołaj Bojańczyk, Sandra Kiefer, and Nathan Lhote.
\newblock String-to-string interpretations with polynomial-size output.
\newblock In {\em 46th International Colloquium on Automata, Languages, and
  Programming, {ICALP} 2019, July 9-12, 2019, Patras, Greece}, pages
  106:1--106:14, 2019.
\newblock Technical report with appendix:
  \url{https://arxiv.org/abs/1905.13190}.
\newblock \href {https://doi.org/10.4230/LIPIcs.ICALP.2019.106}
  {\path{doi:10.4230/LIPIcs.ICALP.2019.106}}.

\bibitem[Boj18]{polyregular}
Mikołaj Bojańczyk.
\newblock Polyregular functions, 2018.
\newblock \href {https://arxiv.org/abs/1810.08760} {\path{arXiv:1810.08760}}.

\bibitem[Boj22]{PolyregSurvey}
Mikołaj Bojańczyk.
\newblock Transducers of polynomial growth (invited talk).
\newblock In Christel Baier and Dana Fisman, editors, {\em {LICS} '22: 37th
  Annual {ACM/IEEE} Symposium on Logic in Computer Science, Haifa, Israel,
  August 2 - 5, 2022}, pages 1:1--1:27. {ACM}, 2022.
\newblock \href {https://doi.org/10.1145/3531130.3533326}
  {\path{doi:10.1145/3531130.3533326}}.

\bibitem[Boj23]{PolyregularGrowth}
Mikołaj Bojańczyk.
\newblock On the growth rates of polyregular functions.
\newblock In {\em 38th Annual ACM/IEEE Symposium on Logic in Computer Science
  (LICS)}, 2023.
\newblock \href {https://doi.org/10.1109/LICS56636.2023.10175808}
  {\path{doi:10.1109/LICS56636.2023.10175808}}.

\bibitem[DT23]{gaetanPhD}
Gaëtan Douéneau-Tabot.
\newblock {\em Optimization of string transducers}.
\newblock PhD thesis, Université Paris Cité, November 2023.
\newblock URL: \url{https://theses.hal.science/tel-04690881}.

\bibitem[GLN25]{StructPoly}
Paul Gallot, Nathan Lhote, and Lê Thành~Dũng
  N{\fontencoding{T5}\selectfont{}guyễn}.
\newblock The structure of polynomial growth for tree automata/transducers and
  mso set queries, 2025.
\newblock To appear in TheoretiCS.
\newblock \href {https://arxiv.org/abs/2501.10270} {\path{arXiv:2501.10270}}.

\bibitem[Lop25]{Aliaume}
Aliaume Lopez.
\newblock {Commutative \(\mathbb{N}\)-Rational Series of Polynomial Growth}.
\newblock In {\em 42nd International Symposium on Theoretical Aspects of
  Computer Science (STACS 2025)}, volume 327 of {\em Leibniz International
  Proceedings in Informatics (LIPIcs)}, pages 67:1--67:16, Dagstuhl, Germany,
  2025. Schloss Dagstuhl -- Leibniz-Zentrum f{\"u}r Informatik.
\newblock \href {https://doi.org/10.4230/LIPIcs.STACS.2025.67}
  {\path{doi:10.4230/LIPIcs.STACS.2025.67}}.

\bibitem[Rab26a]{RabinovichWoLLIC}
Alexander Rabinovich.
\newblock Decidability of mso reparametrization over countable labelled chains,
  2026.
\newblock To appear in the proceedings of WoLLIC 2026.
\newblock \href {https://arxiv.org/abs/2605.18248} {\path{arXiv:2605.18248}}.

\bibitem[Rab26b]{RabinovichICALP}
Alexander Rabinovich.
\newblock From sets to points: Simplifying {MSO} interpretations over countable
  chains, 2026.
\newblock To appear in the proceedings of ICALP 2026.
\newblock \href {https://doi.org/10.4230/LIPIcs.ICALP.2026.163}
  {\path{doi:10.4230/LIPIcs.ICALP.2026.163}}.

\end{thebibliography}

\end{document}